\documentclass[11pt]{article}

\usepackage[utf8]{inputenc}
\usepackage[T1]{fontenc}
\usepackage{a4}
\usepackage{amsmath}
\usepackage{amssymb}
\usepackage{amsthm}
\usepackage{mathrsfs}
\usepackage{xparse}
\usepackage{bbm}
\usepackage{color}
\usepackage{todonotes}
\usepackage{tikz}
\usepackage{authblk}

\DeclareDocumentCommand\setdef{mo}{\left\{#1\IfNoValueTF{#2}{}{ \mid #2}\right\}}
\DeclareDocumentCommand\conv{o}{\operatorname{conv}\IfValueTF{#1}{\left(#1\right)}{}}
\DeclareDocumentCommand\xc{o}{\operatorname{xc}\IfValueTF{#1}{\left(#1\right)}{}}
\DeclareDocumentCommand\R{}{\mathbb{R}}
\DeclareDocumentCommand\Z{}{\mathbb{Z}}
\DeclareDocumentCommand\matroid{o}{\mathcal{M}\IfValueT{#1}{_{#1}}}
\DeclareDocumentCommand\orderO{m}{\mathcal{O}\left(#1\right)}
\DeclareDocumentCommand\zerovec{o}{\IfNoValueTF{#1}{\mathbb{O}}{\mathbb{O}_{#1}}}
\DeclareDocumentCommand\transpose{m}{#1^{\intercal}}
\DeclareDocumentCommand\subgraphPoly{o}{P\IfValueT{#1}{^{#1}}_{\mathrm{sub}}(G)}
\DeclareDocumentCommand\martinPoly{}{P^\star_{\mathrm{sub}}(G)}
\DeclareDocumentCommand\spfPoly{}{P_{\mathrm{sp.forests}}(G)}
\DeclareDocumentCommand\onevec{o}{\IfNoValueTF{#1}{\mathbbm{1}}{\mathbbm{1}_{#1}}}
\DeclareDocumentCommand\countmatroid{}{\mathcal{M}_{m,\ell}(G)}
\DeclareDocumentCommand\sparsitymatroid{}{\mathcal{M}_{k,\ell}(G)}
\DeclareDocumentCommand\matroid{}{\mathcal{M}}
\DeclareDocumentCommand\sperner{}{\mathcal{T}}

\newtheorem{thm}{Theorem}
\newtheorem{prop}[thm]{Proposition}

\title{Subgraph Polytopes and Independence Polytopes of Count Matroids}
\author[1]{Michele Conforti}
\author[2]{Volker Kaibel}
\author[2]{Matthias Walter}
\author[2]{Stefan Weltge}
\affil[1]{Universit\'{a} di Padova: conforti@math.unipd.it}
\affil[2]{Otto-von-Guericke-Universität Magdeburg: \{kaibel,matthias.walter,weltge\}@ovgu.de}
\date{\today}

\begin{document}

\maketitle


\begin{abstract}
Given an undirected graph, the non-empty subgraph
polytope is the convex hull of the characteristic vectors of pairs $(F,S)$ where $S$ is a non-empty subset of nodes
and $F$ is a subset of the edges with both endnodes in $S$.
We obtain a strong relationship between the non-empty subgraph polytope and the spanning forest polytope.
We further show that these polytopes provide polynomial size extended formulations for independence polytopes of count
matroids,
which generalizes recent results obtained by Iwata et al.~\cite{IwataKKKO14} referring to sparsity matroids.
As a byproduct, we obtain new lower bounds on the extension complexity of the spanning forest polytope in terms of
extension complexities of independence polytopes of these matroids.
\end{abstract}


\section{Introduction}

Given an undirected graph $ G = (V,E) $, let us define the \emph{subgraph polytope} of $ G $ as
\[
    \subgraphPoly := \conv \setdef{(\chi(F), \chi(S)) \in \{0,1\}^E \times \{0,1\}^V}[F \subseteq E(S), \, S \subseteq
    V],
\]
where $ \chi(\cdot) $ denotes the characteristic vector and $ E(S) $ denotes the set of edges with both endnodes in $ S
$.
A system of valid linear inequalities whose set of feasible integer points coincides with
the set of integer points in $\subgraphPoly$ is given by
\begin{align}
    \label{eq:trivial1}
    0 \le z_v \le 1 \quad & \forall \, v \in V \\
    \label{eq:trivial2}
    0 \le y_{\{v,w\}} \le z_v \quad & \forall \, \{v,w\} \in E \,.
\end{align}
Let $A$ be the matrix describing system \eqref{eq:trivial1}, \eqref{eq:trivial2}.
Since $A$ has at most one $+1$ and one $-1$ in each row, $A$ is totally unimodular and hence
\begin{multline*}
  \subgraphPoly = \conv\{(\chi(F), \chi(S)) \in \{0,1\}^E \times \{0,1\}^V : \\
                  (\chi(F), \chi(S))\mbox{ satisfies } \eqref{eq:trivial1}, \eqref{eq:trivial2}\} \,.
\end{multline*}

In what follows, we will consider a certain type of subpolytopes (i.e., convex hulls of subsets of
vertices) of $ \subgraphPoly $ and show how they can be used to construct extended formulations for
independence polytopes of count matroids (see Section~\ref{sec:applications}) defined on $ G $.
Let $ \sperner $ be a family of subsets of nodes in $ V $ and consider the convex hull of vertices of $ \subgraphPoly $
that only come from sets $ S $ that contain at least one member of $ \sperner $, i.e., the polytope
\begin{equation}
    \label{eq:typesubpolytopes}
    \subgraphPoly[\sperner] := \conv \setdef{(\chi(F), \chi(S))}[F \subseteq E(S), \, T \subseteq S \subseteq V
    \text{ for some } T \in \sperner].
\end{equation}
For each $ T \in \sperner $, let $ Q_T $ be the face of $ \subgraphPoly $ that is defined by $ x_v = 1$ for all $v \in T $, we clearly
have $ \subgraphPoly[\sperner] = \conv \big( \bigcup_{T \in \sperner} Q_T \big) $.
Except for some discussion at the end of this paper, we will only be concerned with the choice $ \sperner^{\star} = \setdef{\{v\}}[v \in V] $ and the polytope
\[
    \martinPoly := \subgraphPoly[\sperner^{\star}] = \setdef{(\chi(F), \chi(S))}[F \subseteq E(S), \, \emptyset \ne S
    \subseteq V],
\]
which we call the \emph{non-empty subgraph polytope} of $ G $.
Note that $ \martinPoly $ is the convex hull of the vertices of $\subgraphPoly$ that are distinct from $(\zerovec,\zerovec)$.

In this note, we are  interested in  \emph{extended formulations} of polyhedra.
An extended formulation of a polyhedron $ P $ is a system of linear inequalities $ Ax + By \le b $
and equations $ Cx + Dy = d $ such that 
$ P = \setdef{x}[\exists y: Ax + By \le b ,\, Cx + Dy = d] $.
As usual, we only count the number of inequalities, which we define to be the \emph{size} of the extended formulation.
The \emph{extension complexity} $ \xc[P] $ of a polyhedron $ P $ is defined as the smallest size of an extended
formulation for $ P $.
Since  $ \subgraphPoly[\sperner] = \conv \big( \bigcup_{T \in \sperner} Q_T \big) $ and $ \subgraphPoly $
is defined by inequalities \eqref{eq:trivial1} and \eqref{eq:trivial2},
by Balas' extended formulation for the union of polyhedra~\cite{Balas79} we obtain
\begin{equation}
    \label{eq:xcsubpolytopes}
    \xc[\subgraphPoly[\sperner]] \le \sum_{T \in \sperner} \xc(Q_T) + 1 \le \orderO{|\sperner| \cdotp (|V| + |E|)}.
\end{equation}
\medskip

In the first part of this note, we will show a strong relationship between $ \martinPoly $ and the \emph{spanning forest
polytope} of $ G $, which is the convex hull of characteristic vectors of (edge-sets of) forests in $ G $ with the same connected components as $G$ and will be
denoted by $ \spfPoly $.
In Section~\ref{sec:martin}, we revisit the work of Martin~\cite{Martin91}, which shows that any extended formulation
for~$ \martinPoly $ can be transferred into one for~$ \spfPoly $ of nearly the same size.
In Section~\ref{sec:description}, we will provide a complete description of~$ \martinPoly $ in the original space,
which, to our surprise, shows that the converse holds as well:
Any extended formulation for~$ \spfPoly $ can be transferred into one for~$ \martinPoly $ of nearly the
same size and hence the asymptotic growths of their extension complexities coincide.

As our second contribution, we will show in Section~\ref{sec:applications} that $ \martinPoly $ can be used to obtain
polynomial size extended formulations for independence polytopes of special types of count matroids.
Using another subpolytope of type~\eqref{eq:typesubpolytopes}, we will even be able to give polynomial size extended
formulations for the whole class of count matroids.

Recently, Iwata~et~al.~\cite{IwataKKKO14} showed the existence of polynomial size extended formulations for independence
polytopes of \emph{sparsity matroids}, a subclass of count matroids.
They employed a technique developed in~\cite{FaenzaFGT12} and designed a randomized communication protocol (exchanging
only few bits) that computes the slack matrix of these polytope in expectation.
This approach defines an extended formulation only implicitly.
It probably would be a rather tedious task to explicitly derive an extended formulation from that protocol, which consequently is not done in~\cite{IwataKKKO14}.

In our note, we are able to give polynomial bounds on the extension complexities of independence polytopes of count
matroids. In the special case of sparsity matroids, these bounds match the ones obtained in~\cite{IwataKKKO14}.
Our proof technique is completely different and allows to easily work out explicit extended formulations.
In addition, we are even able to improve upon the bounds given in~\cite{IwataKKKO14} in some cases if the underlying
graph is planar.

Although independence polytopes of matroids are arguably well-understood
there are only a few classes of independence polytopes of matroids for which polynomial size extended formulations are known.
On the negative side, Roth\-voss~\cite{Rothvoss13} showed that there exists a family of independence polytopes of matroids
whose extension complexities grow exponentially in their dimension.


\section{Revisiting Martin's construction}
\label{sec:martin}
In this section, we review Martin's~\cite{Martin91} construction of an extended formulation for spanning forest
polytopes in a slightly more abstract manner as given in his paper.
We start with one of his key observations.

\begin{prop}
    \label{prop:separation}
   Given a non-empty polyhedron $Q$ and $\gamma \in \R$, let
    \[
        P = \setdef{x}[\langle x,y \rangle \le \gamma \quad \forall \, y \in Q] \,.
    \]
    If $ Q = \setdef{y}[\exists z: Ay + Bz \le b] $,  we have that
    \[
        P = \setdef{x}[\exists \lambda \ge \zerovec: \transpose{A}\lambda = x, \, \transpose{B}\lambda = \zerovec, \,
        \langle b,\lambda \rangle \le \gamma]
    \]
    holds and hence $ \xc[P] \le \xc[Q] + 1 $.
\end{prop}
\begin{proof}
    A point $ \bar{x} $ is contained in $ P $ if and only if
    \[
        \max \, \setdef{\langle \bar{x}, y \rangle}[\exists z: Ay + Bz \le b] \le \gamma,
    \]
    which by strong duality is equivalent to the existence of dual multipliers $ \lambda \ge \zerovec $ such that $ \transpose{A}\lambda =
    \bar{x} $, $ \transpose{B}\lambda = \zerovec $, and $ \langle b,\lambda \rangle \le \gamma $ hold.
\end{proof}

\noindent
Let $ G = (V,E) $ be an undirected graph and let $ \nu(G) $ be the number of its connected components.
Edmonds~\cite{Edmonds71} shows that $\spfPoly $, the spanning forest polytope of $ G $,
equals the set of points $x \in \R_+^E$ satisfying
$x(E) = |V| - \nu(G)$ and $ x(F) \le |S| - 1$ for all $F \subseteq E(S)$ with $\emptyset \ne S \subseteq V$.
Alternatively, a point $ x \in \R^E_+ $ with $ x(E) = |V| - \nu(G) $ is contained in $ \spfPoly $ if and only if
\[
    \big\langle (x, -\onevec_V) , (\chi(F), \chi(S)) \big\rangle \le -1 \quad \forall \, F \subseteq E(S), \, \emptyset \ne S \subseteq V
\]
holds and hence
\begin{align*}
    \spfPoly = \{ x \in \R^E_+ \ \mid \ & x(E) = \nu(G),\\
    & \big\langle (x, -\onevec_V) , (y,z) \big\rangle \le -1 \quad \forall \, (y,z) \in \martinPoly \}.
\end{align*}
Since $\subgraphPoly$ is defined by the system \eqref{eq:trivial1} and \eqref{eq:trivial2}, by Proposition~\ref{prop:separation} and Inequality~\eqref{eq:xcsubpolytopes} we obtain
\begin{equation}
    \label{eq:forestslemartin}
    \xc[\spfPoly] \le \xc[\martinPoly] + |E| + 1 \le \orderO{|V|(|V| + |E|)}.
\end{equation}
It is an open question whether this bound on the extension complexity of the spanning forest polytope is tight for
general graphs.
One might ask whether the bound given in~\eqref{eq:xcsubpolytopes} is best possible in the case of~$ \martinPoly $.
Note that the extended formulation for~$ \martinPoly $ behind Inequality~\eqref{eq:xcsubpolytopes} is a special case of
those constructed in~\cite{AnguloADK13}, where the general problem of removing vertices from polytopes is investigated.
Clearly, any construction yielding an asymptotically smaller extension for~$ \martinPoly $ would imply an improved upper
bound on the extension complexity of the spanning forest polytope.
In the next section, we will see that also the converse holds.


\section{Description of the non-empty subgraph polytope}
\label{sec:description}

In this section, we first give a complete description of~$ \martinPoly $ in the original space:
\begin{thm}
    \label{thm:subpolyouter}
    For an undirected graph $ G = (V,E) $ we have
    \begin{multline*}
        \martinPoly = \subgraphPoly\cap \{ (y,z) \in \R^E \times \R^V \mid \\
        y(F) \le z(V) - 1 \quad \forall \, F \subseteq E \text{ spanning forest} \}.
    \end{multline*}
\end{thm}
\begin{proof}
    Let $ Q $ denote the polytope on the right-hand side of the equation.
    It is easy to check that the inequalities defining $Q$ imply $z(V) \ge 1$
    and are valid for all vertices of $\subgraphPoly$ except the origin.
    Thus, the integer points in $ \martinPoly $ and $ Q $ coincide
    and it suffices to show that $ Q $ has only integer vertices.

    First, suppose that we have a point $ (y,z) $ that satisfies~\eqref{eq:trivial1} and \eqref{eq:trivial2} with $ z_v = 1
    $ for some $ v \in V $.
    Given a (spanning) forest $ F \subseteq E $, inequalities~\eqref{eq:trivial2} together with nonnegativity of $z$ imply
    \[
        y(F) \le z(V \setminus \{v\}) = z(V) - z_v = z(V) - 1.
    \]
    Thus, every face of $ Q $ defined by $ z_v = 1 $ for some $ v \in V $ coincides with the face of $
    \subgraphPoly $ defined by $ z_v = 1 $ and hence has only integer vertices.

    Let $ (y,z) $ be any vertex of $ Q $.
    It remains to show that this implies $ z_v = 1 $ for some $ v \in V $.
    For the sake of contradiction, assume that we have $ z_v < 1 $ for all $ v \in V $.
    By possibly deleting nodes and edges of $ G $, we may assume that we have $ z_v \ge y_{\{v,w\}} > 0 $ for all $
    \{v,w\} \in E $.
    Then $ (y,z) $ is the unique solution of a system
    \begin{align}
        \label{eq:tightedge}
        y_{\{v,w\}} & = z_v & \text{for all } \{v,w\} \in E' \\
        \label{eq:tightforest}
        y(F) & = z(V) - 1 & \text{ for all } F \in \mathcal{F}
    \end{align}
    of linear equations
    for some $E' \subseteq E $ and some non-empty collection $ \mathcal{F} $ of spanning forests.
    Let $ \alpha := \max_{e \in E} y_e $ and set $ \overline{E} := \setdef{e \in E}[y_e = \alpha] $.
    Let $ (V_\alpha, E_\alpha) $ be a connected component of $ (V, \overline{E}) $ containing at least one edge and let us define 
    $(y',z') \in \R^E \times \R^V $ as follows:
    \begin{align*}
        y'_e &:= \begin{cases}
            2 \cdotp y_e & \text{if } e \in E \setminus E(V_\alpha) \\
            2 \cdotp y_e - 1 & \text{if } e \in E(V_\alpha)
        \end{cases}\\
        z'_v &:= \begin{cases}
            2 \cdotp z_v & \text{if } v \in V \setminus V_\alpha \\
            2 \cdotp z_v - 1 & \text{if } v \in V_\alpha
        \end{cases}
    \end{align*}

    As we have $ y_{\{v,w\}} < z_v $ if $ \{v,w\} \notin E_{\alpha} $ and $ v \in V_{\alpha} $, we obtain that $
    (y',z') $ satisfies~\eqref{eq:tightedge}.
    Let $ F^* $ be a spanning forest such that $ y(F^*) = z(V) - 1 $.
    Since $ y(F) \le z(V) - 1$ for every spanning forest $F$, $F^*$ is a spanning forest of maximum $ y $-weight, following Kruskal's algorithm we find  that 
    $ |F^* \cap E_{\alpha}| = |V_{\alpha}| - 1 $ holds. Hence
    \begin{align*}
        y'(F^*) &= 2y(F^*) - (|V_{\alpha}| - 1) \\
        &= 2(z(V) - 1) - (|V_{\alpha}| - 1) \\
        &= 2z(V) - |V_{\alpha}| - 1 \\
        &= z'(V) - 1.
    \end{align*}
      Since  $ 0 < z_v < 1 $ for all $ v \in V $, $(y',z')\ne  (y,z) $.
      Therefore $ (y',z') $ is another solution to the system~\eqref{eq:tightedge}--\eqref{eq:tightforest} and
      this contradicts the fact that \eqref{eq:tightedge}--\eqref{eq:tightforest} defines a vertex of $Q$.
\end{proof}

Using Proposition~\ref{prop:separation}, the above statement implies that every extended formulation for $ \spfPoly $
can be transferred into one for $ \martinPoly $ of essentially the same size.
\begin{thm}
    \label{thm:forestseqmartin}
    The extension complexities of $ \spfPoly $ and $ \martinPoly $ coincide up to an additive term of order $
    \orderO{|V| + |E|} $.
\end{thm}
\begin{proof}
  By Inequality~\eqref{eq:forestslemartin}, we already have
  $\xc[\spfPoly] \le \xc[\martinPoly] + \orderO{|E|}$.
  Setting
  \begin{align*}
    P &:= \setdef{ (y,z) \in \R^E \times \R^V }[ y(F) - z(V) \leq -1 \quad \forall F \subseteq E \text{ spanning forest} ] \,, \\
    Q &:= \spfPoly \times \setdef{ -\onevec[V] } \,,
  \end{align*}
  and $\gamma := -1$ we obtain
  \begin{align*}
    \xc[\martinPoly] 
      &= \xc[\subgraphPoly \cap P] \\ 
      &\le \xc[\subgraphPoly] + \xc[P] \\
      &\le \xc[\subgraphPoly] + \xc[Q] + 1 \\
      &= \xc[\subgraphPoly] + \xc[\spfPoly] + 1 \\
      &\leq \orderO{ |E| + |V| } + \xc[\spfPoly]
  \end{align*}
  where the first equality follows from Theorem~\ref{thm:subpolyouter},
  the second inequality from Proposition~\ref{prop:separation},
  and the third inequality from the outer description of $\subgraphPoly$
  given by \eqref{eq:trivial1} and \eqref{eq:trivial2}.
\end{proof}


\section{Extended formulations for independence polytopes of count matroids}
\label{sec:applications}

Let $ G = (V,E) $ be an undirected graph.
Given $ \ell \in \Z $, let  $ m \colon V \to \Z_+ $ be a non-negative integer valued function satisfying
\begin{equation}
    \label{eq:countmatroidnontrivial}
    m(v) + m(w) \ge \ell \quad \forall \, \{v,w\} \in E \,.
\end{equation}
Consider the independence system $ \countmatroid $ on ground set $ E $ where a set $ F \subseteq E $ is independent if
and only if
\[
    |F \cap E(S)| \le \max \, \{ m(S) - \ell, 0\},
\]
holds for all $ S \subseteq V $, where $ m(S) = \sum_{v \in S} m(v) $.
Such independence systems can be easily seen to satisfy the matroid axioms and are called \emph{count matroids},
see~\cite{Frank11}.
If we have $ m(v) = k $ for all $ v \in V $ for some $ k \in \Z $, the matroid $ \sparsitymatroid := \countmatroid $ is called a \emph{$
(k,\ell) $-sparsity matroid}.
Note that the $ (1,1) $-sparsity matroid of $ G $ is simply the graphic matroid of $ G $.
A theorem of Nash-Williams~\cite{Nash61} states that the independent sets of the $ (k,k) $-sparsity matroid of $ G $ are
those subsets of edges of~$ E $ that can be partitioned into~$ k $ forests.

In the remainder of this note, we are interested in extended formulations for the \emph{independence polytope} $
P(\matroid) $ of a matroid $ \matroid $, which is defined as the convex hull of characteristic vectors of independent
sets of $ \matroid $.
Recently, Iwata~et~al.~\cite{IwataKKKO14} showed the existence of polynomial size extended formulations for independence
polytopes of $ (k,\ell) $-sparsity matroids.
More precisely, they showed that the extension complexity of $ P(\sparsitymatroid) $ can be bounded by $ \orderO{|V|
\cdotp |E|} $ if $ k \ge \ell $, and by $ \orderO{|V|^2 \cdotp |E|} $ otherwise.
In what follows, we will give bounds on the extension complexities of independence polytopes of count matroids, which
match the ones given in~\cite{IwataKKKO14} for the special case of $ (k,\ell) $-sparsity matroids.

We will distinguish two cases:
In the first case, we assume that $ m, \ell $ satisfy the following additional requirement:
\begin{equation}
    \label{eq:countmatroidsimplecase}
    m(v) \ge \ell \quad \forall \, v \in V
\end{equation}
Note that this case corresponds to the assumption $ k \ge \ell $ in the case of $ (k,\ell) $-sparsity matroids.
The second case deals with the general situation in which~\eqref{eq:countmatroidsimplecase} is not necessarily
satisfied.
\begin{thm}
    \label{thm:countmatroidsimplecase}
    Let $ \countmatroid $ be a count matroid satisfying~\eqref{eq:countmatroidsimplecase}.
    Then we have
    \begin{itemize}
        \item[(a)] $ \xc[P(\countmatroid)] \le \xc[\spfPoly] + \orderO{|V| + |E|} $,
        \item[(b)] $ \xc[P(\countmatroid)] \le \orderO{|V|(|V| + |E|)} $,
        \item[(c)] $ \xc[P(\countmatroid)] \le \orderO{|V| + |E|} $ if $ G $ is planar.
    \end{itemize}
\end{thm}
\begin{proof}
    Due to condition~\eqref{eq:countmatroidsimplecase},
    $ P(\countmatroid) $ can be described via
    \[
        P(\countmatroid) = \setdef{x \in \R^E_+}[x(F) \le m(S) - \ell \quad \forall \, F \subseteq E(S), \, \emptyset \ne S
        \subseteq V]
    \]
    or, alternatively,
    \[
        P(\countmatroid) = \setdef{x \in \R^E_+}[
        \big\langle (x, -m(\onevec_V)) , (y,z) \big\rangle \le -\ell \quad \forall \, (y,z) \in \martinPoly],
    \]
    where $ m(\onevec_V) \in \R^V $ is defined via $ m(\onevec_V)_v = m(v) $ for all $ \forall \, v \in V $.
    Thus, via Proposition~\ref{prop:separation} we conclude
    \[
        \xc[P(\countmatroid)] \le \xc[\martinPoly] + |E| + 1,
    \]
    (the summand $ |E| $ being due to the nonnegativity constraints on $ x $), and hence
    \[
        \xc[P(\countmatroid)] \le \xc[\spfPoly] + \orderO{|V| + |E|},
    \]
    by Theorem~\ref{thm:forestseqmartin},
    which shows (a).
    Part (b) then follows from~\eqref{eq:forestslemartin}.
    For planar graphs $ G $, exploiting the relation between spanning trees in planar graphs and their duals,
    Williams~\cite{Williams02}  showed that $ \spfPoly $ admits an extended formulation of size $
    \orderO{|V| + |E|} $ and hence (c) follows from (a) as well.
\end{proof}
\noindent
In the proof of Theorem~\ref{thm:countmatroidsimplecase} we used the fact that the inequalities
$ 0 \le m(v) - \ell \quad \forall \, v \in V $ are valid for $ P(\countmatroid) $ if $ m,\ell $ satisfy~\eqref{eq:countmatroidsimplecase}
and hence were able to use $ \martinPoly $ to describe $ P(\countmatroid) $.
For the general case, we have to make use of another subpolytope of $ \subgraphPoly $ of
type~\eqref{eq:typesubpolytopes}.
\begin{thm}\label{thm:countGeneral}
    Let $ \countmatroid $ be any count matroid.
    Then we have
    \[
        \xc(P(\countmatroid)) \le \orderO{|E| (|V| + |E|)}.
    \]
\end{thm}
\begin{proof}
    Let us consider the polytope
    \[
        \subgraphPoly[E] = \conv \setdef{(\chi(F), \chi(S))}[F \subseteq E(S), \, e \subseteq S \subseteq V,
        \, e \in E].
    \]
    By Inequality~\eqref{eq:xcsubpolytopes}, we have $ \xc[\subgraphPoly[E]] \le \orderO{|E| (|V| + |E|)} $.
    Due to~\eqref{eq:countmatroidnontrivial},
    $ P(\countmatroid) $ can be described via
    \[
        P(\countmatroid) = \setdef{x \in \R^E_+}[x(F) \le m(S) - \ell \quad \forall \, F \subseteq E(S), \, e \subseteq S
        \subseteq V, \, e \in E]
    \]
    or, alternatively,
    \[
        P(\countmatroid) = \setdef{x \in \R^E_+}[
        \big\langle (x, -m(\onevec_V)) , (y,z) \big\rangle \le -\ell \quad \forall \, (y,z) \in \subgraphPoly[E]].
    \]
    and hence the claim follows from Proposition~\ref{prop:separation}.
\end{proof}

In contrast to the polytope $ \martinPoly $, from computer experiments it seems that the polytope $ \subgraphPoly[E] $ used in the proof of Theorem~\ref{thm:countGeneral} has a very complicated facet structure. In fact, we do not even have a conjecture how an inequality description in the original space could look like.


\bibliographystyle{abbrv}
\bibliography{references}

\end{document}